\documentclass[a4paper,USenglish,cleveref,autoref,thm-restate]{lipics-v2021}

\usepackage{csquotes}
\usepackage{subcaption} 
\usepackage[normalem]{ulem} 

\newtheorem{assumption}[theorem]{Assumption}
\Crefname{assumption}{Assumption}{Assumptions}

\Crefname{paragraph}{Paragraph}{Paragraphs}
\Crefname{subsection}{Subsection}{Subsections}
\Crefname{subsubsection}{Subsection}{Subsections}

    \newcommand{\op}[1]{\operatorname{#1}}
    \newcommand{\var}[1]{\mathit{#1}}

    \newcommand{\set}[1]{\left\{{#1}\right\}}

        \newcommand{\shortFullyRetroactiveMonoPQ}{FR-MonoPQ}
        \newcommand{\shortPartiallyRetroactivePQ}{PR-PQ}
        \newcommand{\shortFullyRetroactivePQ}{FR-PQ}

\pdfoutput=1 
\hideLIPIcs  

\title{Retroactive Monotonic Priority Queues via Range Searching} 

\titlerunning{Retroactive Monotonic Priority Queues} 

\author{Lucas Castro}{Institute of Computing - UFAM, Brazil}{lucas.castro@icomp.ufam.edu.br}{https://orcid.org/0009-0008-0876-823X}{}

\author{Rosiane de Freitas}{Institute of Computing - UFAM, Brazil}{rosiane@icomp.ufam.edu.br}{https://orcid.org/0000-0002-7608-2052}{}

\authorrunning{L. Castro and R. de Freitas} 

\Copyright{Lucas Castro de Souza and Rosiane de Freitas-Rodrigues} 



\category{} 

\relatedversion{}

\funding{This research was partially supported by the Coordination for the Improvement of Higher Education Personnel -- Brazil (CAPES-PROEX) -- Funding Code 001, the National Council for Scientific and Technological Development (CNPq), and the Amazonas State Research Support Foundation -- FAPEAM -- through the POSGRAD 2024-2025 project. Additionally, under Brazilian Federal Law No. 8,387/1991, Motorola Mobility partially sponsored this research through the Artificial Intelligence Techniques for Software Performance Analysis and Optimization (SWPERFI) project, under agreement No. 004/2021, signed with UFAM.}

\acknowledgements{We thank Jonas Costa for his initial contributions. We are also grateful to the anonymous reviewers for suggesting the use of priority search trees.}

\nolinenumbers 

\EventEditors{John Q. Open and Joan R. Access}
\EventNoEds{2}
\EventLongTitle{42nd Conference on Very Important Topics (CVIT 2016)}
\EventShortTitle{CVIT 2016}
\EventAcronym{CVIT}
\EventYear{2016}
\EventDate{December 24--27, 2016}
\EventLocation{Little Whinging, United Kingdom}
\EventLogo{}
\SeriesVolume{42}
\ArticleNo{23}

\begin{document}

\maketitle

\begin{abstract}
The best-known fully retroactive priority queue costs $O(\log^2 m\allowbreak \log \log m)$ time per operation and uses $O(m \log m)$ space, where $m$ is the number of operations performed on the data structure. In contrast, standard (non-retroactive) priority queues can cost $O(\log m)$ time per operation and use $O(m)$ space. So far, it remains open whether these bounds can be achieved for fully retroactive priority queues.

In this work, we study a restricted variant of priority queues known as monotonic priority queues. First, we show that finding the minimum in a retroactive monotonic priority queue is a special case of the range-searching problem. Then, we design a fully retroactive monotonic priority queue that costs $O(\log m)$ time per operation and uses $O(m)$ space, achieving the same bounds as a standard priority queue.

\end{abstract}

\section{Introduction}
Standard data structures operate only in the \enquote{present}. That is, they only take into consideration their current state and \enquote{forget} about their past states. For the vast majority of problems, standard data structures are sufficient and the best option. However, standard data structures are not as versatile as they can be:
\begin{itemize}
    \item Their past states cannot be queried.
    \item Once an element is deleted, all the information about the element is lost.
    \item If an operation is mistakenly performed in the past, it is usually not possible to efficiently undo or modify this operation in the present.
\end{itemize}

Retroactive data structures are a type of data structure that can handle these problems. They were introduced by Demaine et al.~\cite{demaine_retroactive_2007} and are capable of modifying their past states and propagating the consequences of this modification until their present state. Furthermore, they have been used to solve problems such as cloning Voronoi diagrams~\cite{dickerson_cloning_2010} and finding the shortest path in a dynamic graph~\cite{sunita_dynamizing_2021}.

A data structure is called partially retroactive if it can modify the past but can only make queries in the present. A data structure is called fully retroactive if it can modify and query the past.

\paragraph*{Priority queues} Let $m$ be the number of operations performed on the data structure. Demaine et al.~\cite{demaine_retroactive_2007} designed a partially retroactive priority queue with $O(\log m)$ time per operation, matching the bounds of a binary heap~\cite{williams_algorithm_1964}. In contrast, the most efficient fully retroactive priority queue so far has $O(\log^2 m \log \log m)$ time per operation~\cite{dehne_polylogarithmic_2015}. It is known that some fully retroactive data structures need to have a non-constant multiplicative slowdown over their partially retroactive versions~\cite{chen_nearly_2018}, conditioned on widely believed conjectures. However, it is unknown whether this is the case for retroactive priority queues.

A monotonic priority queue is a restricted variant of priority queues, where the extracted elements are restricted to form a non-decreasing function over time. One of its applications is in Dijkstra's algorithm, since the distances of the discovered vertices naturally follow this restriction~\cite{costa_exploring_2025}. Given that a monotonic priority queue exhibits a more restricted behavior, it is usually easier to design it efficiently than to design a more general priority queue efficiently. In this work, we study how to design an efficient fully retroactive monotonic priority queue.

\paragraph*{Our contributions}
\begin{itemize}
    \item We show that finding the minimum in a retroactive monotonic priority queue is a special case of the range-searching problem (\Cref{sec:key_ideas}).
    \item We present a fully retroactive monotonic priority queue that costs $O(\log m)$ time per operation and uses $O(m)$ space (\Cref{sec:data_structure}).
\end{itemize}

\section{Preliminaries} \label{sec:preliminaries}
In this section, we introduce the background necessary to understand the results presented in the next sections. In \Cref{ssec:definitions}, we present the definitions used during this work. In \Cref{ssec:characteristics_background}, we present some important characteristics of retroactive priority queues in order to facilitate the understanding of the topic. Finally, in \Cref{ssec:notation_assumptions}, we list the assumptions and notation used in the rest of this work.

\subsection{Definitions} \label{ssec:definitions}

\begin{definition} \label{def:pq}
    A priority queue can be defined as an abstract data type that maintains a collection of elements and supports the following operations:
    \begin{itemize}
        \item $\op{insert}(x)$: Inserts an element $x$.
        \item $\op{get-min}()$: Returns the minimum element.
        \item $\op{extract-min}()$: Extracts the minimum element.
    \end{itemize}
\end{definition}

\begin{definition}
    A monotonic priority queue is a priority queue with the following (equivalent) constraints: \begin{itemize}
        \item The extracted elements form a non-decreasing sequence~\cite{cherkassky_buckets_1999}. That is, after an element is extracted, only elements greater than or equal to it can be extracted.
        \item Elements smaller than the last-extracted element cannot be inserted~\cite{costa_exploring_2025}.
    \end{itemize}
\end{definition}

\begin{definition} \label{def:rmpq}
    By adding the concept of retroactivity to a priority queue, we can define a retroactive priority queue as an abstract data type with the following operations~\cite{demaine_retroactive_2007}:
    \begin{itemize}
        \item $\op{Insert}(\op{insert}(x), t)$: Inserts an $\op{insert}(x)$ operation at time $t$.
        \item $\op{Delete}(\op{insert}(x), t)$: Deletes an $\op{insert}(x)$ operation at time $t$.
        \item $\op{Insert}(\op{extract-min}, t)$: Inserts an $\op{extract-min}()$ operation at time $t$.
        \item $\op{Delete}(\op{extract-min}, t)$: Deletes an $\op{extract-min}()$ operation at time $t$.
        \item $\op{GetMin}(t)$: Returns the minimum element that exists in the priority queue at time $t$.
    \end{itemize}
\end{definition}

    Note that a partially retroactive priority queue (\shortPartiallyRetroactivePQ{}) has the constraint that the $\op{GetMin}$ operation can be performed only at the present time. That is, $\op{GetMin}(t)$ can be executed only if $t = \infty$. A fully retroactive priority queue (\shortFullyRetroactivePQ{}) does not have this constraint.

\begin{definition} \label{def:monotonic_pq}
    We define a fully retroactive monotonic priority queue (\shortFullyRetroactiveMonoPQ{}) in the same way as an \shortFullyRetroactivePQ{} but with the monotonic constraints applied.
\end{definition}

\subsection{Characteristics of Retroactive Priority Queues} \label{ssec:characteristics_background}

\paragraph*{Ray Shooting Analogy}
There is an analogy between retroactive priority queues and the vertical ray shooting problem~\cite{demaine_retroactive_2007} that can help to understand the behavior of a retroactive priority queue. More specifically, it is possible to represent a retroactive priority queue in a plane in the following way (see \Cref{fig:ray} for an illustration of this representation): 
\begin{itemize}
    \item Let $k$ be the value of an element, $t_{ins}$ be its insertion time, and $t_{del}$ be the time it is extracted. Each inserted element is represented by a segment that goes from point $(t_{ins}, k)$ to point $(t_{del}, k)$. If the element is never extracted from the priority queue, $t_{del} = \infty$.
    \item Let $t_{em}$ be the time that an $\op{extract-min}()$ operation occurs. An $\op{extract-min}()$ operation is a vertical ray that starts at $(t_{em}, -\infty)$ and extends upwards. This ray stops at the first line segment that it intersects. (Note that this segment represents the smallest element that exists at time $t$. Therefore, this ray effectively finds the element extracted by the $\op{extract-min}()$ operation.)
    \item A $\op{GetMin}(t)$ operation can be seen as returning the first segment that a vertical ray starting from $(t, -\infty)$ intersects, without modifying the segment as an $\op{extract-min}()$ operation does.
\end{itemize}

\begin{figure}[htpb]
    \centering
    \begin{subfigure}[t]{0.45\textwidth}
        \centering
        \includegraphics[width=\textwidth]{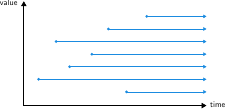}
        \caption{A retroactive priority queue with only insert operations.}
        \label{fig:ray1}
    \end{subfigure}
    \hfill
    \begin{subfigure}[t]{0.45\textwidth}
        \centering
        \includegraphics[width=\textwidth]{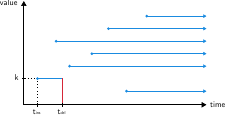}
        \caption{An $\op{extract-min}()$ operation shortens the first intersected segment, forming \enquote{L} shapes.}
        \label{fig:ray2}
    \end{subfigure}

    \bigskip

    \begin{subfigure}[t]{0.45\textwidth}
        \centering
        \includegraphics[width=\textwidth]{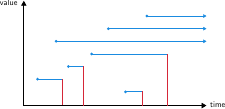}
        \caption{Example with more $\op{extract-min}()$ operations.}
        \label{fig:ray3}
    \end{subfigure}
    \hfill
    \begin{subfigure}[t]{0.45\textwidth}
        \centering
        \includegraphics[width=\textwidth]{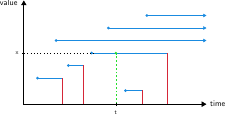}
        \caption{$\op{GetMin}(t)$ operation returning $x$.}
        \label{fig:ray4}
    \end{subfigure}

    \caption{Geometric view of retroactive priority queues. Blue segments represent elements, red segments represent $\op{extract-min}()$ operations, and green dotted segments represent $\op{GetMin}(t)$ queries.}
    \label{fig:ray}
\end{figure}

\paragraph*{Chain Reaction}
Using this representation, it becomes easy to see why a retroactive priority queue can be harder to implement efficiently than a standard one; once a retroactive update is made in the past, the extraction times of $O(m)$ elements can change non-trivially. Thus, in order to propagate these changes, complex approaches are necessary, which can result in poor efficiency for the operations overall. \Cref{fig:chain} shows an example of this. In this example, after the execution of the operation, a subset of elements had their extraction times changed. Therefore, all queries made after this insertion can potentially return a different value than what they would return before.

\begin{figure}[htpb]
    \centering
    \begin{subfigure}[t]{0.45\textwidth}
        \centering
        \includegraphics[width=\textwidth]{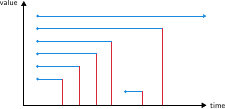}
        \caption{Current representation.}
        \label{fig:chain1}
    \end{subfigure}
    \hfill
    \begin{subfigure}[t]{0.45\textwidth}
        \centering
        \includegraphics[width=\textwidth]{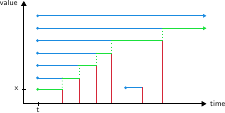}
        \caption{Representation after an $Insert(insert(x), t)$ operation.}
        \label{fig:chain2}
    \end{subfigure}

    \caption{Example of the chain reaction in retroactive priority queues. Lines in green represent the changes that occur because of the operation. Dotted lines represent the state of the segment before the operation. Solid lines represent the state of the segment after the operation.}
    \label{fig:chain}
\end{figure}

Note that in the partially retroactive version, since the past cannot be queried, we do not need to maintain the state of the priority queue at all times. Thus, there is more freedom in how we can propagate the changes caused by the operation. We only need to guarantee that the state of the priority queue is correct in the present. In the fully retroactive version, the ability to query at any point in time makes the process of updating the data structure while maintaining efficient queries more challenging, since it is seemingly necessary to maintain a correct priority queue state at all times.

\paragraph*{Retroactive Monotonic Priority Queues}
We can also use the ray shooting analogy to represent a retroactive monotonic priority queue. We only need to add the monotonic restriction to our previous representation. To do that, we use the second constraint of \Cref{def:monotonic_pq}: Elements smaller than the last-extracted element cannot be inserted.

Let $t_{em}$ be the time that an extraction operation occurs. Let $k$ be the element that was extracted by it. To represent the monotonic restriction, we add a rectangle $(t_{em}, \infty) \times (-\infty, k)$ for each extraction operation. \Cref{fig:mono_ray} shows an illustration of this representation.

In a retroactive monotonic priority queue, it is not possible to insert any new element if its segment starts inside any of these rectangles. This is true because the segment of any element smaller than an already extracted element would start inside them.

Note that this representation is for building intuition only. We will use a different one (in \Cref{ssec:range_search_reduction}) in order to implement retroactive operations efficiently.

\begin{figure}[htpb]
    \centering
    \begin{subfigure}[t]{0.45\textwidth}
        \centering
        \includegraphics[width=\textwidth]{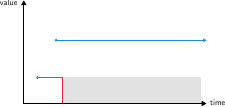}
        \caption{Example with one extraction.}
    \end{subfigure}
    \hfill
    \begin{subfigure}[t]{0.45\textwidth}
        \centering
        \includegraphics[width=\textwidth]{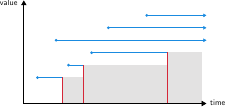}
        \caption{Example with more extractions.}
    \end{subfigure}

    \caption{Illustration of monotonic retroactive priority queues. The gray rectangles represent areas where elements cannot be inserted, because of the monotonic restriction.}
    \label{fig:mono_ray}
\end{figure}

\subsection{Notation and Assumptions} \label{ssec:notation_assumptions}
    The notation used throughout the text is as follows:
    \begin{itemize}
        \item $E$ is the set of elements inserted into the priority queue.
        \item $m$ is the total number of operations performed on the priority queue so far.
        \item We call $val[k]$ the $k$th smallest element inserted into the priority queue. 
        \item We call $em[k]$ the $k$th earliest $\op{extract-min}()$ operation.
        \item $\op{lastExtracted}(t)$ is the last-extracted element before or at time $t$. When there is no extracted value in this interval, $\op{lastExtracted}(t) = -\infty$.
        \item $\op{insertionTime}(x)$ is the time of insertion of an element $x$ into the priority queue.
        \item $\op{extractionTime}(x)$ is the time of extraction of an element $x$ from the priority queue.
    \end{itemize}

We will list the assumptions that we make about the priority queues used in this work to make them explicit and easier to reference during the proofs.

\begin{assumption}[Domain]
    For notational convenience, all elements and times are assumed to be real numbers.
\end{assumption}
However, the presented algorithms work for any totally ordered set that supports constant-time comparison between two elements, without affecting the asymptotic complexities.

\begin{assumption}[Unique Operation Times] \label{def:unique_op_times}
    We assume that each operation on the priority queue occurs at a unique time.
\end{assumption}

\begin{observation} \label{obs:order_em}
    Since all operations occur at different times, we have that $em[i]$ is executed before $em[j]$ for all $i < j$.
\end{observation}

\begin{assumption}[Unique Elements] \label{def:unique_values}
    For simplicity, we assume that all elements inserted into the priority queue are distinct.
\end{assumption}
This assumption can be ignored if we do the following: Let $x$ be the element and $t$ be its insertion time. Instead of referring to $x$ directly, refer to it as $(x, t)$. In this way, all elements become effectively unique again, since $t$ is unique (\Cref{def:unique_op_times}).

\begin{observation} \label{obs:order_val}
    Since all elements are unique, we have that $val[i] < val[j]$ for all $i < j$.
\end{observation}

\begin{assumption}[Consistency] \label{def:valid_ops}
    We assume that only valid operations are performed. Specifically, no $\op{extract-min}()$ operation is performed at a time when the priority queue is empty, and no operation breaks the monotonic property of the priority queue after its execution.
\end{assumption}

\begin{assumption}[Existence Boundaries] \label{def:order_existence}
    We assume that if an element is extracted at time $t$, then it does not exist at time $t$. Also, if an element is inserted at time $t$, then it is present in the priority queue at time $t$.
\end{assumption}

\section{Key Ideas} \label{sec:key_ideas}

In this section, we will establish lemmas about monotonic priority queues that will be helpful to design an \shortFullyRetroactiveMonoPQ{} (fully retroactive monotonic priority queue) in the next section. More specifically, we find a way to efficiently check if an element exists at any time (\Cref{ssec:existence,ssec:fast_last_extracted}) and use that to reduce the query of a retroactive monotonic priority queue to the range-searching problem (\Cref{ssec:range_search_reduction}).

\subsection{Conditions of Existence} \label{ssec:existence}

In order to implement the $\op{GetMin}(t)$ operation, we need to return the minimum element that is present in the priority queue at time $t$. Therefore, $\op{GetMin}(t)$ returns $\min \set{x \mid x \in E,\, x \text{ exists at time } t}$. Since finding the minimum depends on first checking if an element exists, we will focus on finding the minimum later in \Cref{ssec:range_search_reduction}. For now, we will focus on efficiently checking if an element exists.

To check if an element exists at time $t$, we can use the following observations:
\begin{observation} \label{obs:existence_after_insertion}
    After an element is inserted into the priority queue, the only way it can stop existing is by being extracted.
\end{observation}

\begin{observation} \label{obs:existence_range}
    An element $x$ exists only from its $\op{insertionTime}(x)$ until just before its $\op{extractionTime}(x)$. That is, $x$ only exists at times $t \in [\op{insertionTime}(x),\allowbreak\op{extractionTime}(x))$.
\end{observation}

\begin{remark*}
    \Cref{obs:existence_after_insertion} is true because the only operation that modifies the priority queue other than insertion is the $\op{extract-min}()$ operation. \Cref{obs:existence_range} is true because the only way for the elements to stop existing is to be extracted (\Cref{obs:existence_after_insertion}); they exist at the time of insertion but not at the time of extraction (\Cref{def:order_existence}); and all elements are distinct (\Cref{def:unique_values}).
\end{remark*}

Using \Cref{obs:existence_range}, it is possible to check if an element exists at time $t$. So, $\op{GetMin}(t)$ returns \[\min \set{x \mid x \in E,\, \op{insertionTime}(x) \leq t < \op{extractionTime}(x)}.\] However, this is not very useful if we want to support retroactive operations, since one retroactive operation could change the extraction time of $O(m)$ elements, and we would need to update all of them (see \Cref{ssec:characteristics_background}).

Therefore, we need a way to check for element existence that depends on values that are easy to maintain when a retroactive operation occurs. Because of that, we show the following:

\begin{lemma}[Conditions for Existence] \label{lem:existence}
    In a monotonic priority queue, an element $x$ exists at time $t$ iff $\op{insertionTime}(x) \leq t$ and $x > \op{lastExtracted}(t)$.
\end{lemma}
\begin{proof}
    We first prove the forward direction, and then the backward direction. For both directions, assume it is a monotonic priority queue.

    \uline{Forward Direction (Necessary Condition):}

    If $x$ exists at time $t$, then $\op{insertionTime}(x) \leq t$. This is true because if $x$ was inserted after time $t$, it would not exist at time $t$ yet (\Cref{obs:existence_range}).
    
    Assume $x$ exists at time $t$. Assume for a contradiction that $x \leq \op{lastExtracted}(t)$. There are two cases ($=$ and $<$). We show that a contradiction occurs in both.
    
    \textit{Case 1: $x = \op{lastExtracted}(t)$}. Since $\op{lastExtracted}(t)$ was extracted, it does not exist at time $t$. And $x$ exists at time $t$. They have the same value, which implies that they are the same element (\Cref{def:unique_values}). However, one exists at time $t$ while the other does not, which is a contradiction.
    
    \textit{Case 2: $x < \op{lastExtracted}(t)$}. Since, at time $t$, $\op{lastExtracted}(t)$ was extracted and $x$ was not, $x$ is extracted after a greater element was extracted, contradicting that it is a monotonic priority queue.
    
    Therefore, if $x$ exists at time $t$, then $x > \op{lastExtracted}(t)$. This completes the forward direction.

    \uline{Backward Direction (Sufficient Condition):}
    
    Assume that $x > \op{lastExtracted}(t)$. Assume that $\op{insertionTime}(x) \leq t$.
    
    Assume for a contradiction that $x$ does not exist at time $t$. By \Cref{obs:existence_after_insertion}, the only way for it to not exist at time $t$ is if $\op{extractionTime}(x) \leq t$. However, this contradicts the assumption that it is a monotonic priority queue: Since $x > \op{lastExtracted}(t)$, $x$ is not $\op{lastExtracted}(t)$ and is greater than it. This implies that $x$ was extracted at some moment and then $\op{lastExtracted}(t)$ was extracted after, creating a decreasing function over time. This completes the backward direction, and hence the proof.
\end{proof}

Directly from \Cref{lem:existence}, we have a new way of finding the minimum element in the priority queue at any time.

\begin{corollary}[GetMin Expression] \label{cor:get_min_expression}
    The $\op{GetMin}(t)$ operation returns
    \begin{equation} \label{eq:getmin}
            \min\,\{\,x \mid x \in E,\, \op{insertionTime}(x) \le t,\, x > \op{lastExtracted}(t)\,\}.
    \end{equation}
\end{corollary}

\subsection{Finding the Last Extracted Element} \label{ssec:fast_last_extracted}
Now we show a way to efficiently find $\op{lastExtracted}(t)$. For that, we can maintain the list of $\op{extract-min}()$ operations in a binary search tree and find the last extraction until time $t$, via a predecessor search on $t$. However, we have no way yet of finding which element this operation extracts. Thus, we prove the following lemma:
\begin{lemma} \label{lem:shift}
    In a monotonic priority queue, $val[k]$ can only be extracted by $em[k]$.
\end{lemma}
\begin{proof}
    
    This lemma will be proven in two parts. For both parts, assume it is a monotonic priority queue.

    \uline{Part 1: $val[k]$ can be extracted by $em[k]$.}
    
    Assume $val[k]$ is extracted by $em[k]$. By \Cref{obs:order_val}, $val[k] < val[l]$ for all $k < l$. By \Cref{obs:order_em}, $em[k]$ is executed before $em[l]$ for all $k < l$.
    
    If $val[k]$ is extracted by $em[k]$ for all possible $k$, then $val[k]$ is extracted before $val[l]$ for all $k < l$. Therefore, the extracted values create an increasing function as $k$ increases (that is, as time passes). This satisfies the monotonic priority queue restriction, showing that it is possible for this to happen.

    \uline{Part 2: if $j \neq k$, $val[k]$ cannot be extracted by $em[j]$.}

    For $j$ to be different from $k$, $j$ needs to be greater or less than $k$. Therefore, we prove this part by these two cases.

    \textit{Case 1: $j < k$.} Assume for a contradiction that $val[k]$ is extracted by $em[j]$ and $j < k$. By the pigeonhole principle, there is an element smaller than $val[k]$ that needs to be extracted by $em[l]$ and $l \geq k$. Since this element is smaller than $val[k]$ and is being extracted after $val[k]$, the extracted values form a decreasing sequence at some point, which contradicts the assumption that it is a monotonic priority queue.

    \textit{Case 2: $j > k$.} Assume for a contradiction that $val[k]$ is extracted by $em[j]$ and $j > k$. By the pigeonhole principle, there is an element greater than $val[k]$ that needs to be extracted by $em[l]$ and $l \leq k$.  Since this element is greater than $val[k]$ and is being extracted before $val[k]$, the extracted values form a decreasing sequence at some point, which contradicts the assumption that it is a monotonic priority queue.

    Since the $k$th element cannot be extracted by any $j$ smaller or greater than $k$, $j$ cannot be different from $k$. Thus, this part of the proof is complete.

    In Part 1, we proved that $val[k]$ can be extracted by $em[k]$. In Part 2, we proved that it cannot be extracted by any other $\op{extract-min}()$ operation. Therefore, $val[k]$ can only be extracted by $em[k]$.
\end{proof}

Since we can now find which element each $\op{extract-min}()$ operation extracts, efficiently finding $\op{lastExtracted}(t)$ becomes straightforward.
\begin{lemma} \label{cor:fast_last_extracted}
    In a monotonic priority queue, there is a procedure that, given $t$, finds $\op{lastExtracted}(t)$ using $O(\log m)$ time.
\end{lemma}
\begin{proof}
    Let all inserted elements be stored in an order-statistic tree $T_{el}$~\cite[Sec.~17.1]{cormen_introduction_2022}. Let all extraction times be stored in an order-statistic tree $T_{em}$. To find $\op{lastExtracted}(t)$, we can execute the following algorithm:
    \begin{itemize}
        \item[(1)] Find the predecessor of $t$ in $T_{em}$, calling it $\var{lastOp}$.
        \item[(2)] Find the order of $\var{lastOp}$ in $T_{em}$, calling it $k$.
        \item[(3)] Find the $k$th smallest element of $T_{el}$, calling it $x$.
        \item[(4)] Return $x$.
    \end{itemize}

    We claim that $x = \op{lastExtracted}(t)$. The proof is as follows. In Step 1, we find $\var{lastOp}$, the last extraction operation until time $t$. (Note that $\var{lastOp}$ extracts $\op{lastExtracted}(t)$. Hence, we need to return the element that $\var{lastOp}$ extracts.) In Step 2, we find its order $k$. Hence, $\var{lastOp} = em[k]$. We know that $\var{lastOp} = em[k]$ extracts $val[k]$ (\Cref{lem:shift}). Thus, in Step 3, we find $x = val[k]$. Since $x$ is the element extracted by $\var{lastOp}$, $x = \op{lastExtracted}(t)$.

    Finding the predecessor of a value, the order of a value, and the $k$th smallest element in an order-statistic tree can be done in $O(\log n)$ time each, where $n$ is the number of elements in it. Since the trees store all operations, they can have at most $m$ elements. Therefore, $n \leq m$, and the algorithm takes $O(\log m)$ time in total.
\end{proof}

\begin{remark*}
    We will not use it, but we can find $\op{extractionTime}(x)$ in $O(\log m)$ time in a similar way: (1) Find the order of $x$, calling it $k$; (2) find the operation that extracts $x$ (by finding $em[k]$); (3) return the time that $em[k]$ is performed.
\end{remark*}

\subsection{Reduction to Range Searching} \label{ssec:range_search_reduction}
Now that we know that $\op{GetMin}(t)$ is equivalent to Expression~\ref{eq:getmin} and that we can efficiently find $\op{lastExtracted}(t)$, we just need to find the minimum among all elements that exist at time $t$.

To do that efficiently, we will represent the monotonic priority queue geometrically. The idea comes directly from Expression~\ref{eq:getmin}. We plot all the elements and their insertion times as points in a 2D plane. Once we do that, we observe that the two conditions of existence of an element at time $t$ (\Cref{lem:existence}) correspond to a rectangle intersection. The $\op{GetMin}(t)$ query then becomes finding the lowest point that intersects the rectangle. \Cref{fig:getmin} shows an example of this representation. Here is a detailed specification:

\begin{lemma} \label{lem:reduction}
    We can represent a monotonic priority queue in the following way: All elements $e$ in $E$ become points $(\op{insertionTime}(e), e)$ in a plane. 

    Let $R(t)$ be the two-sided rectangle $(-\infty, t] \times (\op{lastExtracted}(t), \infty)$. Let $p$ be the point with the minimum $y$-coordinate that intersects $R(t)$. In this representation, the $\op{GetMin}(t)$ query returns the $y$-coordinate of $p$.
\end{lemma}
\begin{proof}
    For a point $(\op{insertionTime}(e),\allowbreak e)$ to intersect $R(t)$, it needs to have $\op{insertionTime}(e) \leq t$ and $e > \op{lastExtracted}(t)$, which are exactly the conditions of existence of an element at time $t$ (\Cref{lem:existence}). Therefore, all points that intersect $R(t)$ represent elements that exist at time $t$. And since the $y$-coordinate of a point stores the element it represents, the point with the minimum $y$-coordinate that is inside the rectangle represents the smallest element that exists at time $t$.
\end{proof}

\begin{figure}[htpb]
    \centering
    \includegraphics[width=0.6\textwidth]{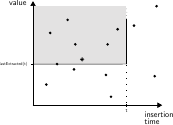}
    \caption{Illustration of \Cref{lem:reduction}. The point that intersects the rectangle is $\op{lastExtracted}(t)$. All points inside the shaded area (the rectangle) exist at time $t$, except $\op{lastExtracted}(t)$. The marked point is the point with minimum value among all the points that exist at time $t$.}
    \label{fig:getmin}
\end{figure}

The problem of getting information about the set of points within a given rectangle is a problem called \textit{range searching}. This is a well-studied problem with many variations, each with its own specialized data structures (see~\cite{goodman_range_2018} for a survey on the subject).

For our variation, we need to be able to insert and delete points (retroactive insertion). We also need to find the point with the minimum $y$-coordinate that intersects a given two-sided orthogonal rectangle ($\op{GetMin}(t)$ query). This can be achieved by using priority search trees~\cite{mccreight_priority_1985}, which support all these operations in $O(\log n)$ time and $O(n)$ space, where $n$ is the number of points stored in the data structure.

\section{The Data Structure} \label{sec:data_structure}
Using what we know from the previous section, we can now design an efficient \shortFullyRetroactiveMonoPQ{}.

\begin{theorem} \label{thm:rmpq}
    There is an \shortFullyRetroactiveMonoPQ{} with the following costs:
    \begin{itemize}
        \item $O(\log m)$ time per $\op{Insert/Delete}(\op{extract-min}, t)$ operation.
        \item $O(\log m)$ time per $\op{Insert/Delete}(\op{insert}(x), t)$ operation.
        \item $O(\log m)$ time per $\op{GetMin}(t)$ operation.
        \item $O(m)$ space.
    \end{itemize}
\end{theorem}
\begin{proof}
    We utilize three auxiliary data structures $T_{el}$, $T_{em}$, and $T_{ins}$:
    \begin{itemize}
        \item $T_{el}$: An order-statistic tree that stores all inserted elements.
        \item $T_{em}$: An order-statistic tree that stores all the times that an $\op{extract-min}()$ operation occurs.
        \item $T_{ins}$: A priority search tree that stores all inserted elements and their insertion times. More specifically, let $(t, x)$ be an ordered pair that represents an element $x$ inserted at time $t$. $T_{ins}$ stores the pairs $(t, x)$ as points.
    \end{itemize}

    \textit{Operations}. We first design the $\op{GetMin}(t)$ operation. As we showed in \Cref{ssec:range_search_reduction}, the $\op{GetMin}(t)$ operation can be solved by a query on a priority search tree. Hence, we simply query $T_{ins}$. But to do that, we first have to find $\op{lastExtracted}(t)$ to create the rectangle. Accordingly, this is how we can implement {\boldmath$\op{GetMin}(t)$}:

    \begin{itemize}
        \item[(1)] Find $\op{lastExtracted}(t)$.
        \item[(2)] Query $T_{ins}$ with the rectangle $(-\infty, t] \times (\op{lastExtracted}(t), \infty)$. Let $p$ be the returned point.
        \item[(3)] Return the $y$-coordinate of $p$.
    \end{itemize}

    Next, we design the other operations. Since $T_{el}$ and $T_{ins}$ store the elements, they need to be updated every time an $\op{Insert}/\op{Delete}(\op{insert}(x), t)$ operation is performed. Accordingly, this is how we can implement {\boldmath$\op{Insert}(\op{insert}(x), t)$}:
    \begin{itemize}
        \item[(1)] Insert $x$ into $T_{el}$.
        \item[(2)] Insert the pair $(t, x)$ into $T_{ins}$.
    \end{itemize}
    
    The $\op{Delete}(\op{insert}(x), t)$ operation is implemented similarly, but it removes the element instead of inserting it.
    
    Since $T_{em}$ stores all extraction times, it needs to be updated every time an $\op{Insert}/\allowbreak\op{Delete}(\op{extract-min}, t)$ operation is performed. Accordingly, this is how we can implement {\boldmath$\op{Insert}(\op{extract-min}, t)$}: Insert $t$ into $T_{em}$. The $\op{Delete}(\op{extract-min}, t)$ operation is implemented similarly, but it removes $t$ instead of inserting it.

    \textit{Time analysis}. Every operation involves only a constant number of calls to a priority search tree or order-statistic tree. Therefore, each operation takes $O(\log m)$ time.
    
    \textit{Space analysis}. Since each operation adds at most one element to each data structure, each data structure can have at most $m$ elements. Thus, $T_{el}$, $T_{em}$, and $T_{ins}$ use $O(m)$ space in total.
\end{proof}

\section{Conclusions}
In this paper, we showed that finding the minimum in a retroactive monotonic priority queue is a special case of the range-searching problem. Then, we designed an \shortFullyRetroactiveMonoPQ{} with $O(\log m)$ time per retroactive operation and $O(m)$ space. This result is optimal in the comparison model.

In future work, we intend to study the upper bound of other restricted versions of priority queues with full retroactivity, the upper bound of the general fully retroactive priority queue and/or the lower bound of retroactive priority queues.

\bibliography{paper}

\end{document}